\documentclass[11pt]{article}
\usepackage{amsmath,amssymb,amsthm}
\usepackage[final]{showkeys}
\usepackage{mathbbol,bm,bbm}

\newcommand{\Cx}{{\mathbb C}}

\newcommand{\Rl}{{\mathbb R}}

\newcommand{\idty}{\Eins}
\DeclareMathOperator{\id}{id}

\newcommand{\co}[1]{\textsf{#1}}

\DeclareMathOperator{\tr}{Tr}
\newcommand{\<}{\langle}
\renewcommand{\>}{\rangle}
\providecommand{\abs}[1]{\lvert#1\rvert}
\providecommand{\norm}[1]{\lVert#1\rVert}

\renewcommand{\c}[1]{\mathcal{#1}}

\newcommand{\s}[1]{\mathsf{#1}}
\renewcommand{\r}[1]{\mathrm{#1}}
\newcommand{\jam}{Jamio\l kowski}
\renewcommand{\det}{\r{Det}\,}

\newtheorem{lemma}{Lemma}
\newtheorem{proposition}{Proposition}
\newtheorem{definition}{Definition}
\newtheorem{theorem}{Theorem}
\newtheorem{corollary}{Corollary}

\begin{document}

\begin{center}
{\LARGE  Composition of quantum states and \\[12pt] 
dynamical subadditivity} \\[12pt]
W.~Roga$^1$, M.~Fannes$^2$ and K.~{\.Z}yczkowski$^{1,3}$
\end{center}

\medskip
\noindent
$^1$ Instytut Fizyki im.~Smoluchowskiego,
Uniwersytet Jagiello{\'n}ski,
PL-30-059 Krak{\'o}w, Poland \\
$^2$ Instituut voor Theoretische Fysica,
Universiteit Leuven, B-3001 Leuven, Belgium \\
$^3$Centrum Fizyki Teoretycznej, Polska Akademia Nauk,
PL-02-668 Warszawa, Poland

\medskip
\noindent
Email: \texttt{<wojciechroga@wp.pl>}, \texttt{<mark.fannes@fys.kuleuven.be>}, and \\ \texttt{<karol@tatry.if.uj.edu.pl>} 

\bigskip
\noindent
\textbf{Abstract:}
We introduce a composition of quantum states of a bipartite system which is based on the reshuffling of density matrices. This non-Abelian product is associative and stems from the composition of quantum maps acting on a simple quantum system. It induces a semi-group in the subset of states with maximally mixed partial traces. Subadditivity of the von~Neumann entropy with respect to this product is proved. It is equivalent to subadditivity of the entropy of bistochastic maps with respect to their composition, where the entropy of a map is the entropy of the corresponding state under the \jam\ isomorphism. Strong dynamical subadditivity of a concatenation of three bistochastic maps is established. Analogous bounds for the entropy of a composition are derived for general stochastic maps. In the classical case they lead to new bounds for the entropy of a product of two stochastic matrices.


\section{Introduction}

General quantum dynamics are described by quantum stochastic maps, also called \co{quantum channels} or \co{quantum operations}. It is therefore crucial to investigate their properties in order to understand admissible dynamics in a quantum state space. The \co{\jam\ isomorphism}~\cite{Ja72,ZB04} associates with any quantum stochastic operation a quantum state acting on an extended space. Up to normalization, this state is nothing else than the \co{dynamical or Choi matrix} of the map. Thus the features of a quantum map are encoded in a state.

The spectral decomposition of the \jam\ state yields a canonical Kraus decomposition of the map, i.e.\ realizes the map in terms of measurement operators. For a unitary evolution, the corresponding state is pure and the Kraus form consists of a single unitary operator. For the completely depolarizing channel, the state is maximally mixed and the Kraus decomposition consists of many terms. The degree of mixing of the measurement operators required to construct a quantum map by its canonical Kraus form can therefore be estimated by a quantity like the \co{entropy of the map} which is actually the entropy of the associated state. Thus this entropy vanishes for a unitary evolution and reaches its maximal value for the completely depolarizing channel.

The aim of this paper is to analyze properties of the composition of stochastic maps~\cite{WR07}. Using the subadditivity of entropy for composite systems we prove an analogous \co{dynamical subadditivity} for bistochastic maps, i.e.\ maps which preserve the identity, see~(\ref{addquant1}). A similar inequality proved for the concatenation of three bistochastic maps may be called \co{strong dynamical subadditivity}~(\ref{strongb}). Dynamical subadditivity generalizes to general stochastic maps by adding an extra term which vanishes for bistochastic maps. Restricting to diagonal states we obtain bounds for the entropy of the product of two classical stochastic matrices, see~(\ref{addquant2b}). This generalizes a recent result of S\l omczy\'nski on entropy of a product of bistochastic matrices~\cite{Sl02}.

Composition of quantum maps induces an action in the space of quantum states on a bipartite system. We analyze properties of this action which allows to construct a semi-group in the space of Hermitian matrices. Composition of states induces also a semi-group structure in the set of positive definite operators on a bipartite system whose partial traces are proportional to the identity.

The main tools are coupling techniques, which associate to quantum maps states on composite systems, combined with subadditivity and strong subadditivity of quantum entropy. The explicit constructions that we use are related to similar constructions of Lindblad and techniques used in studying quantum dynamical entropy in the sense of~\cite{alifan}. There are also connections with quantum coherent information and inequalities as the quantum information data processing inequality, see~\cite{SN96}.       

To illustrate the composition of states in action we work out the case of quasi-free Fermionic maps. In this setup we derive an explicit form of the composition and discuss dynamical subadditivity.

The paper is organized as follows. In Section~II the necessary properties of quantum states and quantum maps are reviewed. In particular we consider the case of quantum operations with a diagonal dynamical matrix and show that any such stochastic or bistochastic quantum map reduces to a stochastic or bistochastic matrix acting on classical probability vectors. The notion of composition of states is introduced in Section~III where some of its properties are analyzed. In Section~IV we analyze the entropy of maps and formulate dynamical and strong dynamical subadditivity for compositions of bistochastic maps. Furtheremore we discuss analogous results in a more general case of stochastic maps. More explicit examples of low-dimensional quasi-free structures are presented in the Appendix.

\section{Quantum states and quantum maps}

\subsection{Quantum dynamical matrices}

Let $\rho$ denote a $N$-dimensional density matrix i.e.\ a Hermitian, positive operator, satisfying the trace normalization condition $\tr \rho = 1$. The expectation value of an observable $X$, i.e.\ a $N$-dimensional matrix, is given by the usual relation
\begin{equation}
 \<X\> = \tr \rho\,X \ . 
\end{equation}  
Such expectation functionals are called states. We shall in the sequel identify $\<\ \>$ with its corresponding $\rho$. Let $\c D_N$ denote the set of generally mixed quantum states acting on the $N$ dimensional Hilbert space $\c H_N$.  It is a convex, compact set of real dimensionality $N^2-1$. In the case of a qubit, i.e.\ $N=2$, the space of mixed states is the Bloch ball, $\c D_2 = B_3 \subset \Rl^3$. More generally, there is a one to one correspondence between arbitrary, i.e.\ not necessarily positive linear functionals $F$ on the observables and $N$-dimensional matrices $\rho$ 
\begin{equation}
 F(X) = \tr \rho\,X \ .
\end{equation} 

General quantum maps, sometimes called super-operators, are linear transformations either of the observables (Heisenberg picture) or of the functionals on the observables (Schr\"odin\-ger picture). In this paper we shall mostly use the Schr\"odinger picture and denote such quantum maps by $\Phi$. The adjoint $\Phi^\dagger$ of a super-operator $\Phi$ is given by 
\begin{equation}
 \Phi^\dagger(\rho) = \bigl( \Phi(\rho^*) \bigr)^* .
\end{equation}
Fixing a basis in $\c H_N$ we identify a $N$-dimensional matrix $\rho$ with a vector of dimension $N^2$ just by writing the entries of $\rho$ in lexicographical order. So the entry $\rho_{m\mu}$ is placed on the $\bigl( (m-1)N + \mu \bigr)$-th row. A general linear quantum map, 
\begin{equation}
 \rho \mapsto \rho' := \Phi(\rho)
\end{equation} 
may be described by a matrix of size $N^2$ still denoted by $\Phi$,
\begin{equation}
 \rho'_{m\mu}  
 = \Phi_{\substack{m\mu\\n\nu}}\, \rho_{n\nu} \ ,
\label{dynmatr1}
\end{equation}
where Einstein's summation convention is taken.

Another convenient way to describe quantum maps is to use the \co{Choi-\jam\ encoding} or \co{dynamical matrix} $D_\Phi$~\cite{SMR61}. It amounts to a reordering of matrix elements of $\Phi$
\begin{equation}
 D_{\Phi} \equiv \Phi^{\s R}
 \qquad\text{so that}\qquad
 \bigl( D_\Phi \bigr)_{\substack{mn\\\mu\nu}} 
 = (\Phi^{\s R})_{\substack{mn\\\mu\nu}} 
 = \Phi_{\substack{m\mu\\n\nu}} \ .
\label{dynmatr3}
\end{equation}
Let us consider the projector on the maximally entangled state $|\psi^+\> := \frac{1}{\sqrt N} \sum_{m=1}^N |m\> \otimes |m\>$
\begin{equation}
 P_+ 
 := \bigl| \psi^+ \bigr\> \bigl\< \psi^+ \bigr| 
 = \frac{1}{N}\, \sum_{m\mu} |m\>\<\mu| \otimes |m\>\<\mu|\ .
\end{equation}
We can in a similar way as above identify $P_+$ with a vector in a space of dimension $N^2 \times N^2$. Its entries are
\begin{align}
 \bigl( P_+ \bigr)_{mn\,\mu\nu} 
 &= \frac{1}{N}\, \sum_{k\kappa} 
 \bigl( |k\> \<\kappa| \otimes |k\>\<\kappa| \bigr)_{mn\,\mu\nu}
 = \frac{1}{N}\, \sum_{k\kappa} 
 \bigl( |k\> \<\kappa| \bigr)_{m\mu}\, \bigl( |k\>\<\kappa| \bigr)_{n\nu}
\nonumber \\
 &= \frac{1}{N}\, \sum_{k\kappa} 
 \delta_{km}\, \delta_{\kappa\mu}\, \delta_{kn}\, \delta_{\kappa\nu}
 = \frac{1}{N}\, \delta_{mn}\, \delta_{\mu\nu}\ .
\end{align} 
We now compute
\begin{align}
 \Bigl( \bigl( \Phi \otimes \id \bigr)(P_+) \Bigr)_{mn\,\mu\nu}
 &= \bigl( \Phi \otimes \id \bigr)_{\substack{mn\,\mu\nu\\m'n'\,\mu'\nu'}}\,
 \bigl( P_+ \bigr)_{m'n'\,\mu'\nu'} 
\nonumber \\
 &= \Phi_{\substack{m\mu\\m'\mu'}}\, \delta_{nn'}\, \delta_{\nu\nu'}\,
 \bigl( P_+ \bigr)_{m'n'\,\mu'\nu'} 
\nonumber \\
 &= \frac{1}{N}\, \Phi_{\substack{m\mu\\m'\mu'}}\, \delta_{nn'}\, \delta_{\nu\nu'}\,
 \delta_{m'n'}\, \delta_{mu'\nu'} 
\nonumber \\
 &= \frac{1}{N}\, \Phi_{\substack{m\mu\\n\nu}} = \bigl( D_\Phi \bigr)_{\substack{mn\\\mu\nu}}
\end{align}
So we see that, up to a factor $N$, $D_\Phi$ is the action of $\id \otimes \Phi$ on the one-dimensional projection $P_+$ on the maximally entangled state. The map $\Phi \mapsto D_\Phi$ is linear and it intertwines adjoints
\begin{equation}
 D_{\Phi^\dagger} = \bigl( D_\Phi \bigr)^*.
\end{equation}  
Equation~(\ref{dynmatr3}) may be considered as the definition of the \co{reshuffling transformation}, written $\Phi \to \Phi^{\s R}$, which is defined for any matrix $\Phi$ acting on the Hilbert space $\c H_{N^2} = \c H_N \otimes \c H_N$~\cite{ZB04}. It should be stressed that the reshuffling operation depends on the distinguished basis in $\c H_N$ that we have used.

Choi's theorem~\cite{Cho75a}, proves that a map $\Phi$ is \co{completely positive} (CP), which means that the extended map $\Phi\otimes {\id}$ is positive for any size of the extension, if and only if the corresponding dynamical matrix $D_\Phi$, also called \co{Choi matrix}, is positive, $D_\Phi\ge 0$. The eigenvalue decomposition of the dynamical matrix of a CP map $\Phi$ leads to the canonical Kraus form~\cite{Kr71} of the map
\begin{equation}
 \Phi(\rho) = \sum_{\alpha=1}^{N^2} A_\alpha\rho A_\alpha^{\dagger} \ .
\label{Kraus}
\end{equation}
where the Kraus operators are orthogonal 
\begin{equation}
 \langle A_\alpha|A_\beta\rangle
 =\tr A_\alpha^\dagger A_\beta 
 = d_\alpha \delta_{\alpha\beta}\; , 
\end{equation}
so that the non-negative weights $d_\alpha$ become the eigenvalues of the dynamical matrix $D_\Phi$. Hence, in this (almost) canonical form the number of Kraus operators does not exceed $N^2$. 

A quantum map $\Phi$ is \co{trace preserving} (TP) if $\tr \Phi(\rho) = \tr \rho$ for any $\rho$.  The corresponding dynamical matrix $D_\Phi$ acts on the composite Hilbert space $\c H_{N^2} = \c H_A \otimes \c H_B$ and, in terms of the dynamical matrix, trace preserving means that
\begin{equation}
 \tr_A D_\Phi  =  \idty \ .
\label{partrace}
\end{equation}
This implies in particular that $\tr D_\Phi = N$. Completely positive trace preserving maps (CPTP maps) are often called \co{quantum operations or quantum stochastic maps}. Since the dynamical map of a quantum operation $\Phi$ is positive and normalized as in~(\ref{partrace}) the rescaled matrix $\frac{1}{N}\, D_\Phi$ is a state on the extended Hilbert space $\c H_N \otimes \c H_N$, see~\cite{Ja72,BZ06}. We shall say that
\begin{equation}
 \varsigma := {\textstyle\frac{1}{N}}\, D_\Phi
\label{jam}
\end{equation}
is the \co{\jam\ state} associated to $\Phi$. A quantum map is called \co{unital} if it leaves the maximally mixed state invariant. This is the case iff
\begin{equation}
 \tr_B D_\Phi  = \idty \ ,
\label{partraceB}
\end{equation}
a condition dual to~(\ref{partrace}). A CP quantum map which is trace preserving and unital is called \co{bistochastic}. A composition of bistochastic maps is still bistochastic.

\subsection{Classical case - diagonal dynamical matrices.}

Let us diagonalize a density matrix. The elements on the diagonal may be interpreted as a classical probability vector $P$ with components $p_i = \rho_{ii}$, i.e.\ $p_i \ge 0$ and $\sum_{i=1}^N p_i = 1$.  In a similar way a quantum map $\Phi$ reduces to a classical one if the dynamical matrix $D = \Phi^R$ is diagonal, $D_{\substack{ab\\cd}} = T_{ab}\delta_{ac}\delta_{bd}$.
 
Reshaping the diagonal of the dynamical matrix which has dimension $N^2$ one obtains a matrix $T$ of dimension $N$
\begin{equation}
 T = T(\Phi) 
 \qquad\text{with}\qquad 
 T_{ij} = \Phi_{\substack{ii\\jj}},
\label{tij} 
\end{equation} 
with no summation performed. Positivity of $D$ implies that all elements of $T$ are non negative.  Furthermore, the partial trace condition~(\ref{partrace}) implies that the matrix $T$ is stochastic since $\sum_{i=1}^N T_{ij} = 1$ for all $j=1,2,\ldots,N$. In fact, on diagonal matrices the action of a diagonal dynamical matrix $D$ reduces to a Markov transition of a probability vector, $P'=TP$.

If, additionally, the complementary partial trace condition~(\ref{partraceB}) holds, then the matrix $T$ is bistochastic $\sum_{j=1}^N T_{ij} = 1$ for all $i=1,2,\ldots, N$ and the uniform vector $P_* := (\frac{1}{N}, \ldots, \frac{1}{N})$ is invariant under multiplication by a bistochastic matrix $T$. Hence quantum stochastic and bistochastic maps acting in the space $\c D_N$ of quantum states can be considered as a direct generalizations of stochastic and bistochastic matrices, which act on classical probability vectors.

Note that the quantum identity map, $\Phi = \id$, is not classical, since the dynamical matrix $D = \Phi^{\s R}$ is not diagonal, and the off--diagonal elements of $\rho$ are preserved. On the other hand, the coarse graining map,
\begin{equation}
 \Phi_{\r{CG}}(\rho) := \sum_{i=1}^N |i\> \< i |\rho |  i\> \<i| 
\label{coarse}
\end{equation}
which strips away all off diagonal elements of a density operator is described by a diagonal dynamical matrix.  The corresponding stochastic matrix is the identity, $T\bigl(\Phi_{\r{CG}}\bigr) = \idty$, since all diagonal elements of $\rho$ remain untouched under the action of~(\ref{coarse}). Let us also distinguish the flat stochastic matrix $T_*$ whose elements are all equal, $(T_*)_{ij} = \frac{1}{N}$.  It is a bistochastic matrix which maps any probability vector $P$ into the uniform one, $T_*P=P_*$.

\subsection{Quasi-free Fermionic states and maps}

Quasi-free states and maps on a Fermionic algebra will be used as an example. More details and references to the original papers can be found in~\cite{alifan} and~\cite{evakaw}. In appendix we shall describe in more detail these objects for systems with few modes. 
 
The algebra $\c A(\c H_N)$ of observables of a $N$-mode Fermionic system is generated by creation and annihilation operators $a^*(\varphi)$ and $a(\varphi)$ where $\varphi$ belongs to the one-particle space $\c H_N$. These operators satisfy the canonical anti-commutation relations (CAR)
\begin{equation}
 a^*(\varphi + \alpha\psi) = a^*(\varphi) + \alpha a^*(\psi)\ , \quad
 \{a^*(\varphi), a^*(\psi)\} = 0
 \qquad\text{and}\qquad
 \{a^*(\varphi), a(\psi)\} = \<\psi,\varphi\> \ .
\end{equation}  
It is not hard to see that $\c A(\c H_N)$ is isomorphic to the algebra of matrices of dimension $2^N$. The space $\c H_N$ is called the one-particle space.

There exists a widely used class of states and maps called quasi-free. These objects correspond quite literally to Gaussian objects for Fermionic systems and they are fully characterized by operators on the one-particle space. Essentially every effective description of Fermionic systems, such as the Hartree-Fock approximation, is based on quasi-free structures.

\begin{proposition}[Quasi-free states]
Let $Q$ be a $N$-dimensional matrix. Every element of $\c A(\c H_N)$ can be written as a linear combination of \co{ordered monomials}, i.e.\ monomials in creation and annihilation operators where the $a^*$'s appear to the left of the $a$'s. Define a linear functional $\<\ \>_Q$ on $\c A(\c H_N)$ by extending linearly its definition on ordered monomials 
\begin{align}
 &\bigr\< a^*(\varphi_1) \cdots a^*(\varphi_n) a(\psi_n) \cdots a(\psi_1)
 \bigr\>_Q
 := \det \Bigl( \bigl[  \<\psi_j \,,\, Q\,\varphi_i\> \bigr]\Bigr) 
\\
\intertext{and}
 &\< \r{mon} \>_Q = 0 \text{ for every other ordered monomial $\r{mon}$.}
\end{align}
Then $\<\ \>_Q$ is a state on $\c A(\c H_N)$ iff $0 \le Q \le \idty$, i.e.\ $Q$ is Hermitian and every eigenvalue of $Q$ belongs to $(0,1)$. States of the type $\<\ \>_Q$ are called \co{quasi-free} and $Q$ is called the \co{symbol} of the state.
\end{proposition}

The set of quasi-free states on $\c A(\c H_N)$ is not convex, so it is natural to consider the convex hull of the quasi-free states which is a proper subset of the full state space. It is known that any quasi-free state can be decomposed into a convex combination of quasi-free states whose symbols are projectors. Moreover, a quasi-free state is pure iff its symbol is a projection operator. Therefore the extreme points of the convex hull of quasi-free states consist of the states with a projector as symbol. It is remarkable that the extreme points of the convex set of symbols $(0,\idty)$ precisely consists in the projection operators on $\c H_N$. The maximally mixed state, i.e.\ the normalized trace on $\c A(\c H_N)$ is quasi-free and given by $Q = \frac{1}{2}\, \idty$.

Let $\c A(\c H_N \oplus \c H_N)$ be a CAR algebra for a composite system, it is isomorphic to the tensor product of two copies of $\c A(\c H_N)$. It is however, because of the Fermionic nature of the system, more natural to consider the graded tensor product where each of the two parties is identified by the maps
\begin{equation}
 a^*(\varphi_1) \hookrightarrow a^*(\varphi_1 \oplus 0)
 \qquad\text{and}\qquad
 a^*(\varphi_2) \hookrightarrow a^*(0 \oplus \varphi_2) \ .
\end{equation}    
The quasi-free state with symbol 
\begin{equation}
 {\textstyle\frac{1}{2}}\, \begin{pmatrix} 
 \idty&\idty\\\idty&\idty \end{pmatrix}
\label{mmqf}
\end{equation}
is pure on $\c A(\c H_N \oplus \c H_N)$ and its restriction to each of the parties $\c A(\c H_N)$ is maximally mixed, therefore it is a maximally entangled state on the composite system.

Quasi-free CPTP maps are conveniently described in Heisenberg picture by their explicit action on monomials~\cite{eva,fanroc}. For the purpose of this paper it will suffice to know their action on quasi-free states. 

\begin{definition}[Quasi-free CPTP maps]
A CPTP quasi-free map is a CPTP map on the state space of $\c A(\c H_N)$ which maps quasi-free states in quasi-free states. 
\end{definition}

It can be shown~\cite{dirfanpog} that quasi-free TPCP maps are determined by a couple of $N$-dimensional matrices $R$ and $Z$ which satisfy the constraint
\begin{equation}
 0 \le Z \le \idty - R^*R \ .
\label{qfcp}
\end{equation}
Moreover, the action on quasi-free states is given by
\begin{equation}
 \<\ \>_Q  \mapsto \<\ \>_{R^*Q\,R + Z}\; .
\end{equation}
It is easily seen that the constraints~(\ref{qfcp}) are necessary and sufficient to guarantee that for every symbol $Q$ the expression $R^*Q\,R + Z$ is again a symbol. In order to mark the dependence of the map on $R$ and $Z$ we shall use the notation $\Phi_{R,Z}$. The map $\<\ \>_Q  \mapsto \<\ \>_{R^*Q\,R + Z}$ is affine on $(0,\idty)$ but not linear but we may use the standard representation
\begin{equation}
 Q \approx \begin{pmatrix} Q\\\idty \end{pmatrix}
 \qquad\text{and}\qquad
 \Phi_{R,Z} \sim 
 \begin{pmatrix} \r{Ad}(R)&Z\\0&\id\end{pmatrix} \ .
\label{phiqf}
\end{equation}
Here $\r{Ad}(R)$ is the adjoint map $X \mapsto R^*X\,R$. Composing $\Phi_{R_2,Z_2} \circ \Phi_{R_1,Z_1}$ returns the map $\Phi_{R_3,Z_3}$ with
\begin{equation}
 R_3 = R_2 R_1
 \qquad\text{and}\qquad
 Z_3 = R_2^* Z_1 R_2 + Z_2 \ .
\end{equation}
In terms of the representation~(\ref{phiqf}) this amounts to the usual matrix multiplication
\begin{equation}
 \begin{pmatrix} \r{Ad}(R_3)&Z_3\\0&\id\end{pmatrix}
 = \begin{pmatrix} \r{Ad}(R_2)&Z_2\\0&\id\end{pmatrix} 
 \begin{pmatrix} \r{Ad}(R_1)&Z_1\\0&\id\end{pmatrix} \ .
\end{equation}
The representation~(\ref{phiqf}) is the quasi-free analogue of~(\ref{dynmatr1}).

Using~(\ref{mmqf}) we see that the \jam\ state associated with $\Phi_{R,Z}$ is quasi-free on $\c A(\c H_N \oplus \c H_N)$ with symbol
\begin{equation}
 {\textstyle\frac{1}{2}}\, \begin{pmatrix} \idty&R\\R^*&R^*R + 2Z \end{pmatrix} \ .
\label{dmqf}
\end{equation}
The set of quasi-free CPTP maps is again not convex. A quasi-free map $\Phi_{R,Z}$ is extreme within the set of CPTP maps iff 
\begin{equation}
 Z = \sqrt{\idty - R^*R}\, P\, \sqrt{\idty - R^*R}
\label{eqfm}
\end{equation}
with $P$ a projection operator. The same construction that applies for quasi-free states allows to decompose any quasi-free CPTP map into a mixture of extreme quasi-free CPTP maps. Therefore the extreme points of the convex hull of quasi-free CPTP maps consists precisely of the maps of the form~(\ref{eqfm}). 

Finally, a quasi-free CPTP map $\Phi_{R,Z}$ is bistochastic iff
\begin{equation}
 Z = {\textstyle\frac{1}{2}}\, (\idty - R^*R) \ . 
\end{equation}
This follows from the invariance of the tracial state which has symbol $\frac{1}{2}\, \idty$. As such a map is fully determined by a single $N$-dimensional matrix $R$ with $\norm R \le 1$, we simply write $\Phi_R$.

\subsection{Entropies of maps and states}

To any normalized probability vector $P$ of size $N$ we may associate its \co{Shannon entropy}
\begin{equation}
 \s H(P) := - \sum_{i=1}^N p_i \ln p_i = \sum_{i=1}^N \eta(p_i) \ ,
\label{shannon}
\end{equation}
where we have introduced the function
\begin{equation}
 \eta(x) := -x \ln x \text{ for } x>0 
 \qquad\text{and}\qquad
 \eta(0) := 0 \ . 
\end{equation}
This entropy is a measure for the mixedness of a probability vector. In a similar way the degree of mixing of a quantum state $\rho$ is characterized by its \co{von~Neumann entropy}
\begin{equation}
 \s S(\rho) := - \tr \rho \ln \rho = \tr \eta(\rho) \ ,
\label{neumann}
\end{equation}
equal to the Shannon entropy of its spectrum. The entropy varies from zero for a pure state to $\ln N$ for the maximally mixed state, $\rho_* = \frac{1}{N}\, \idty$.

The density matrix $\varsigma = D_\Phi/N$ associated to a quantum stochastic map $\Phi$ depends on the basis that is used to compute the entries of $\Phi$, see~(\ref{dynmatr1}). A change of basis in $\c H_N$ corresponds to a unitary transformation of $D_\Phi$, therefore the eigenvalues of $D_\Phi$ don't change. It is then natural to consider the von~Neumann entropy of the bipartite state associated with $\Phi$ by the \jam\ isomorphism.

\begin{definition}
Let $\Phi$ be a trace-preserving completely positive map with associated Jamio\l kow\-ski state $\varsigma = \frac{1}{N}\, D_\Phi$. The \co{entropy of the map $\Phi$} is defined to be
\begin{equation}
 \s S(\Phi) 
 := \s S(\varsigma) 
 = \s S\bigl( {\textstyle\frac{1}{N}} D_\Phi \bigr) \ .
\label{neumann2}
\end{equation}
\end{definition}
Since the state $\varsigma = D_\Phi/N$ acts on the extended Hilbert space $\c H_N \otimes \c H_N$ the entropy of the map varies from zero for a unitary dynamics to $2 \ln N$ for a completely depolarizing channel $\Phi_*$ which sends any state to the maximally mixed state, $\Phi_*(\rho) = \rho_*$, see~\cite{ZB04}.
 
Let us now move to a classical discrete dynamics in the probability simplex. The following definition of entropy of a stochastic matrix introduced in~\cite{Sl02,ZKSS03}
\begin{equation}
 \s H_{\r I}(T) := \sum_{j=1}^N  p^{\r I}_j\, \s H(t_{j}) \ ,
\label{shanTI}
\end{equation}
is decorated by a label ``I'', as it is based on an invariant state of a matrix, $P^{\r I} = TP^{\r I}$. Here $t_j$ denotes the $j$--th column of a transition matrix $T$, so~(\ref{shanTI}) represents the average Shannon entropy of columns of $T$ weighted by its invariant state $P^{\r I}$.

To demonstrate a direct relation to the quantum dynamics we shall use a simplified version of the entropy of a transition matrix
\begin{equation}
 \s H(T) := - {\textstyle\frac{1}{N}}\, \sum_{i=1}^N  \sum_{j=1}^N T_{ij}
 \ln T_{ji} \ .
\label{shanT}
\end{equation}
Observe that for any bistochastic matrix the uniform vector is invariant, $P^{\r I} = P_*$, so $p^{\r I}_j = \frac{1}{N}$ and both definitions of entropy do coincide. Both quantities, $\s H_{\r I}(T)$ and $\s H(T)$, vary from zero to $\ln N$.

Using eq.~(\ref{tij}) one concludes that for any stochastic map $\Phi$ represented by a diagonal dynamical matrix $D_\Phi$ its entropy is up to a constant equal to the entropy of the associated stochastic matrix,
\begin{equation}
 \s S(\Phi) = \s H\bigl( T(\Phi)\bigr)+ \ln N \ . 
\label{sphist}
\end{equation}
The constant $\ln N$ is due to the $\frac{1}{N}$ normalization factor in front of the dynamical matrix, see definition~(\ref{neumann2}). Although in general the entropy of a quantum map belongs to $[0,2\ln N]$, the entropy of the maps represented by diagonal $D$ and corresponding to stochastic matrices vary from $\ln N$ to $2 \ln N$.  Among this class the minimal entropy characterizes the coarse graining map~(\ref{coarse}), for which $T(\Phi_{\r{CG}}) = \idty$, and $\s S(\Phi_{\r{CG}}) = \ln N$. The maximum is achieved for the completely depolarizing channel, $\Phi_*$, since $T(\Phi_*) = T_*$ so that $\s S(\Phi_{*}) = \s H(T_*) + \ln N = 2\ln N$.

The von~Neumann entropy of a Fermionic quasi-free state with symbol $Q$ can easily be expressed in terms of $Q$
\begin{equation}
 \s S^{\r{qf}}(Q) = \tr \bigl( \eta(Q) + \eta(\idty - Q) \bigr) \ .
\label{entqfs}
\end{equation}
Quite explicit expressions can be given for the entropies of extreme CPTP and bistochastic quasi-free maps using~(\ref{dmqf}) and~(\ref{eqfm}).

For an extreme CPTP quasi-free map, we obtain after some algebraic manipulations
\begin{align}
 \s S(\Phi_{R,Z}) 
 &= \s S^{\r{qf}} \bigl( {\textstyle\frac{1}{2}}
 (\idty + \abs R^2 - 2\abs R\, P\, \abs R) \bigr) 
\nonumber \\
 &= \tr \Bigl( \eta\bigl( {\textstyle\frac{1}{2}}
 (\idty + \abs R^2 - 2\abs R\, P\, \abs R) \bigr) 
 + \eta\bigl( {\textstyle\frac{1}{2}}
 (\idty - \abs R^2 + 2\abs R\, P\, \abs R) \bigr)\Bigr) \ .
\end{align}
In these expressions $P$ is a projector and $Z = \sqrt{\idty - \abs R^2}\, P\, \sqrt{\idty - \abs R^2}$. 

For a bistochastic map determined by $R$ with $\norm R \le 1$, we obtain
\begin{equation}
 \s S(\Phi_R)  
 = 2 \tr \Bigl( \eta\bigl( {\textstyle\frac{1}{2}}(\idty + \abs R) \bigr) 
 + \eta\bigl( {\textstyle\frac{1}{2}}(\idty - \abs R) \bigr) \Bigr) \ .
\end{equation}

\subsection{Entropy exchange and Lindblad's theorem}

Consider a CP map $\Phi$ represented in its canonical Kraus form~(\ref{Kraus}). For any state $\rho\in \c D_N$ define a positive operator $\hat\sigma = \hat\sigma(\Phi,\rho)$ acting on the extended Hilbert space $\c H_{N^2}$,
\begin{equation}
 \hat\sigma_{\alpha\beta} 
 := \tr \rho A_\beta^\dagger A_\alpha , \quad \alpha,\beta=1,\ldots, N^2 \ .
\label{osigma}
\end{equation}
If the map $\Phi$ is stochastic, then the operator $\hat\sigma$ is also normalized in the sense that
\begin{equation}
 \tr \hat\sigma 
 = \sum_{\alpha=1}^{N^2} \tr \rho A_\alpha^\dagger A_\alpha
 = \tr \rho = 1
\end{equation}
and so it represents a density operator in its own right, $\hat\sigma \in \c D_{N^2}$. In particular, if $\rho = \rho_* = \frac{1}{N}\, \idty$ then, using the canonical Kraus decomposition~(\ref{Kraus}), one shows that 
\begin{equation}
 \hat\sigma(\Phi,\rho_*) = \frac{1}{N}\, D_\Phi = \varsigma\; .
\end{equation} 
The von~Neumann entropy of $\hat\sigma$ depends on $\rho$, and equals $\s S(\Phi)$, as defined above, if $\rho$ is the maximally mixed state.

Auxiliary states $\hat\sigma$ in an extended Hilbert space were used by Lindblad to derive bounds for the entropy of the image $\rho' = \Phi(\rho)$ of an initial state under the action of a CPTP map. Lindblad's bounds~\cite{Li91}
\begin{equation}
 \abs{\s S(\hat\sigma) - \s S(\rho)} 
 \le \s S(\rho') 
 \le \s S(\hat\sigma) + \s S(\rho) \ ,
\label{boundent}
\end{equation}
are obtained by defining yet another density matrix in the composite Hilbert space $\c H_N \otimes \c H_M$
\begin{equation}
 \omega 
 := \sum_{\alpha=1}^M \sum_{\beta=1}^M A_\alpha \rho A_\beta^\dagger 
 \otimes |\alpha\>\<\beta| \ ,
 \label{omega}
\end{equation}
where $M = N^2$ and $\{|\alpha\>\}$ is an orthonormal basis in $\c H_M$. Computing partial traces one finds that 
\begin{equation}
 \tr_N \omega = \hat\sigma 
 \qquad\text{and}\qquad 
 \tr_M \omega = \rho' \ .
\end{equation} 
It is possible to verify that $\s S(\omega) = \s S(\rho)$, and so one arrives at~(\ref{boundent}) using subadditivity of the entropy and the triangle inequality~\cite{AL70}. These results can be obtained using the first part of the proof of Theorem~\ref{thm1}.

If the initial state is pure, that is if $\s S(\rho) = 0$, we find that the final state $\rho'$ has entropy $\s S(\hat\sigma)$. For this reason $\s S(\hat\sigma)$ was called the \co{entropy exchange of the operation $\Phi$} by Shumacher~\cite{Sc96}. In that work an alternative representation of the entropy exchange was given
\begin{equation}
 \s S\bigl( \hat\sigma(\Phi,\rho) \bigr) 
 = \s S\bigl( (\id \otimes \Phi) |\varphi\> \<\varphi| \bigr) \ ,
\label{osigma2}
\end{equation}
\noindent 
where $|\varphi\>$ is an arbitrary purification of the mixed state, $\tr_B |\varphi\> \<\varphi| = \rho$. To prove this useful relation it is enough to find a pure state in an extended Hilbert space, such that one of its partial traces gives $\hat \sigma$ and the other one the argument of the entropy function at the right hand side of~(\ref{osigma2}).

A kind of classical analogue of the quantum entropy bound~(\ref{boundent}) of Lindblad was proved later by S{\l}omczy{\'n}ski~\cite{Sl02}. He introduced the notion of entropy of a stochastic matrix $T$ with respect to some fixed stationary probability distribution $P =\{p_i\}_{i=1}^N$, 
\begin{equation}
 \s H_P(T) := \sum_{i=1}^N p_i\, \s H(\vec{t}_i) \,  \qquad 
 \vec{t}_i = (T_{1i}, T_{2i}, \dots , T_{Ni}) \ . 
\label{SPTsum}
\end{equation}
This quantity --- an average entropy of columns of matrix $T$ weighted by probability vector $P$ --- allows to obtain the bounds for the entropy of a classically transformed state, $P'=TP$,  
\begin{equation}
 \s H_P(T) \le \s H(P') \le \s H_{\r P}(T) + \s H(P) \ . 
\label{boundclass}
\end{equation}
These bounds look somewhat similar to the quantum result~(\ref{boundent}) of Lindblad, but a careful comparison is required. Applying the definition~(\ref{osigma}) of the auxiliary state $\hat \sigma$ to the classical case of a diagonal state, $\rho_{ii}=p_i \delta_{ij}$ and a diagonal dynamical matrix $D$ we find
\begin{equation}
 \s S(\sigma^{\r{class}}) = \s H_P(T) + \s H(P) \ , 
\label{ssigma}
\end{equation}
where $T = T(D^{\s R})$ is the classical stochastic matrix given by~(\ref{tij}). Substituting this result into the Lindblad bound~(\ref{boundent}), and renaming $\rho$ and $\rho'$ into $P$ and $P'$ we realize that the argument of the absolute value in the lower bound reduces to $\s H_P(T)$ and is not negative, so we arrive at
\begin{equation}
 \s H_P(T) \le \s H(P') \le \s H_P(T) + 2\s H(P) \ . 
\label{boundclass2}
\end{equation}
The lower bound coincides exactly with the result~(\ref{boundclass}) of S{\l}omczy{\'n}ski. The upper bound is weaker (note the presence of the term $2\s H(P)$ instead of $\s H(P)$), but it holds in general for all quantum maps, while~(\ref{boundclass}) is true for classical dynamics only.

\section{Composition of maps and composition of states}

We first recall some properties of the reshuffling transformation of a matrix as defined in~(\ref{dynmatr3}). Reshuffling does not preserve the spectrum nor the Hermiticity of a matrix. It is an involution, since performing this transformation twice returns the initial matrix, $(X^{\s R})^{\s R} = X$.

Using the \jam\ isomorphism the composition of maps acting on a single quantum system can be used to define a \co{composition} between quantum states of a bi-partite system, in fact, this composition extends to arbitrary matrices.

\begin{definition}
The reshuffling operation~(\ref{dynmatr3}) defines a composition between  arbitrary matrices $\sigma_1$ and $\sigma_2$ on a composite system  $\c H_N \otimes \c H_N$ 
\begin{equation}
 \sigma_1 \odot \sigma_2 
 := \bigl( \sigma_1^{\s R} \sigma_2^{\s R}\bigr)^{\s R}  \ .
\label{compos}
\end{equation}
\end{definition}

For stochastic matrices obtained by reshaping two diagonal density matrices $T_1 = T(\sigma_1^{\s R})$ and $T_2 = T(\sigma_2^{\s R})$ according to~(\ref{tij}), the composition of the diagonal states returns the usual multiplication of stochastic matrices, 
\begin{equation}
 T\bigl( (\sigma_1 \odot \sigma_2)^{\s R} \bigr)
 = T(\sigma_1^{\s R})\, T(\sigma_2^{\s R}) = T_1 T_2 \ .
\label{T1T2}
\end{equation}

Using the definition of reshuffling we see that generally $(D^{\s R})^2$ differs from $(D^2)^{\s R}$. Therefore the composition performed on two copies of a state $\sigma$ differs from its square 
\begin{equation}
 \sigma^{\odot 2} := \sigma \odot \sigma \ne \sigma^2  \ .
\label{sigma22}
\end{equation}

\subsection{Properties of the composition}

\begin{lemma}
Let $X$ and $Y$ denote two matrices of size $N^2$.
\begin{equation}
 \text{If } X \ge 0 \text{ and } Y\ge 0 \text{ then }
 X \odot Y = (X^{\s R} Y^{\s R})^{\s R} \ge 0 \ .
\label{reshlem}
\end{equation}
\end{lemma}

\begin{proof}
Denoting by $\Psi$ and $\Phi$ the completely positive maps with corresponding dynamical matrices $X$ and $Y$ we see that $X \odot Y$ is the dynamical  matrix of the composed map $\Psi \circ \Phi$. Since the composition of two completely positive maps yields again a completely positive map~\cite{BZ06}, we infer~(\ref{reshlem}).
\end{proof}

\begin{proposition}
The set of all operators acting on a composite Hilbert space $\c H_N \otimes \c H_N$, equipped with the composition law $(\odot)$, is a non-Abelian associative semi-group. Moreover, if $\sigma_1$ and $\sigma_2$ are Hermitian operators on $\c H_N \otimes \c H_N$ then also $\sigma_1 \odot \sigma_2$ is Hermitian. Therefore the set of all Hermitian operators on $\c H_N \otimes \c H_N$ is a non-Abelian associative subsemi-group.
\end{proposition}

\begin{proof}
Let $P_+$ denote the projector on the maximally entangled state. Since $(P_+)^{\s R} = \idty$, this operator plays the role of the neutral element of the composition. The composition $\odot$ is non-Abelian, $\sigma_1 \odot \sigma_2 \ne \sigma_2 \odot \sigma_1$, because the composition of quantum maps is not commutative. It is on the other hand associative $(\sigma_1 \odot \sigma_2) \odot \sigma_3 = \sigma_1 \odot (\sigma_2 \odot \sigma_3)$, since the composition of maps is.

It remains to prove is that $\sigma_1 \odot \sigma_2$ is Hermitian if $\sigma_1$ and $\sigma_2$ are. Let $\Phi$ be a super-operator associated with $D_\Phi$ through the \jam\ isomorphism, then $\bigl( D_\Phi \bigr)^* = D_{\Phi^\dagger}$. Moreover, for two super-operators $\Phi$ and $\Psi$ we have
\begin{equation}
 (\Phi \circ \Psi)^\dagger(X) 
 = \bigl( \Phi \circ \Psi (X^*) \bigr)^*
 = \bigl( \Phi (\Psi(X^*)) \bigr)^*
 = \Phi^\dagger \bigl( \bigl( \Psi(X^*) \bigr)^* \bigr)
 = \bigl( \Phi^\dagger \circ \Psi^\dagger \bigr)(X)
\end{equation} 
Using
\begin{equation}
 D_{\Phi \circ \Psi} = D_\Phi \odot D_\Psi
\end{equation}
finishes the proof.
\end{proof}

Restricting our attention to the set $\c D_{N^2}$ of quantum states on a composite system, we see that this algebraic structure breaks down since the trace condition, $\tr\sigma = 1$, is not preserved under composition. However, one may overcome this difficulty by selecting a certain subset of quantum states. Thus consider the subset $\c D^{\r I}_{N^2}$ of density matrices of a composite system $\c H_A \otimes \c H_B$ of size $N^2$ such that their partial trace over the first system is the maximally mixed state
\begin{equation}
 \c D^{\r I}_{N^2} :=  \{ \sigma \in \c D_{N^2} \,:\, 
 \tr_A \sigma = {\textstyle\frac{1}{N}}\, \idty \} \ .
\label{substoch}
\end{equation}
With respect to the \jam\ isomorphism these states  correspond to the trace preserving maps. Since a composition of any two trace preserving maps preserves the trace we infer that the composition~(\ref{compos}) acts internally in the set $\c D^{\r I}_{N^2}$.

\begin{proposition}
The set of all operators acting on the Hilbert space $\c H_N$ of a bipartite system such that their left marginal is proportional to the identity, equipped with the composition $\odot$ is a non-Abelian associative semi-group.
\end{proposition}

It is convenient to distinguish another composition sub-algebra by defining the set of states with both marginals proportional to the identity,
\begin{equation}
 \c D^{\r{II}}_{N^2} := \{ \sigma \in \c D^{\r I}_{N^2} \,:\, 
 \tr_B \sigma = {\textstyle\frac{1}{N}}\, \idty \} \ .
\label{subbist}
\end{equation}
Due to condition~(\ref{partraceB}) this semi-group is generated by compositions of bistochastic maps. 

\subsection{Idempotent states}

Consider the state $\sigma$ of a bipartite system obtained by extending an arbitrary state $\rho$ by the maximally mixed state
\begin{equation}
 \sigma := {\textstyle\frac{1}{N}}\, \rho \otimes \idty \ .
\label{extmix}
\end{equation}
This state is  proportional to the dynamical matrix $D$ of the operation $\Phi_\rho$, which acts as a complete single-step contraction, sending any initial state $\omega$ into $\rho$,
\begin{equation}
 \sigma = {\textstyle\frac{1}{N}}\, D_{\Phi_\rho}\ ,
 \qquad\text{where}\qquad
 \Phi_\rho(\omega) := \rho
\label{extend2}
\end{equation}
for any $\omega \in \c D_N$. To show this let us start with the dynamical matrix of this map, $D_{\substack{mn\\\mu\nu}} = \rho_{m\mu}\, \delta_{n\nu}$. Writing out the matrix entries of 
\begin{equation}
 \omega' := \Phi_\rho(\omega) = D^{\s R} \omega = (\rho \otimes
 \idty)^{\s R} \omega
\end{equation}
in the standard basis we obtain the desired result
\begin{equation}
 \omega'_{m\mu} = D_{\substack{mn\\\mu\nu}}\, \omega_{n\nu} =
 \rho_{m\mu} (\tr \omega) =  \rho_{m\mu}\ .
\end{equation}

A state $\sigma$ of an algebra is called \co{idempotent} if $\sigma \odot \sigma = \sigma$. Hence any state of the form~(\ref{extmix}) is idempotent, since the single-step contracting map $\Phi_\rho$ applied a second time does not influence the system anymore, $\Phi_\rho \circ \Phi_\rho = \Phi_\rho$.

\subsection{Composing quasi-free quantum maps}

The \jam\ state associated to a quasi-free quantum map $\Phi_{R,Z}$ has the symbol
\begin{equation}
 0 \le {\textstyle\frac{1}{2}}\, \begin{pmatrix} \idty&R\\R^*&R^*R + 2 Z\end{pmatrix}
 \le \idty \ .
\label{sym}
\end{equation}
When $R$ and $Z$ vary through the matrices of dimension $N$ subject to the constraint $0 \le Z \le \idty - R^*R$ the symbol~(\ref{sym}) varies trough the matrices of dimension $2N$ in the interval $(0,\idty)$ which are of the form
\begin{equation}
 {\textstyle\frac{1}{2}}\, \begin{pmatrix} \idty&R_1\\R_1^*&R_2\end{pmatrix} \ .
\label{sym2}
\end{equation}
Remark that this set is convex. Composing quasi-free quantum maps induces now the following composition law $\odot$ on these matrices
\begin{align}
 &{\textstyle\frac{1}{2}}\, \begin{pmatrix} \idty&R_1 \\ R_1^*&R_2 \end{pmatrix}
 \odot {\textstyle\frac{1}{2}}\, \begin{pmatrix} \idty&S_1 \\S_1^*&S_2 
 \end{pmatrix}
 =: {\textstyle\frac{1}{2}}\, \begin{pmatrix} \idty&T_1 \\T_1^*&T_2 \end{pmatrix} 
\nonumber \\
\intertext{with}
 &T_1 = R_1 S_1
 \qquad\text{and}\qquad
 T_2 = S_1^* (R_2 - \idty) S_1 + S_2 \ . 
\label{qfodot} 
\end{align}
This composition law is the analogue of~(\ref{compos}). It is associative and non-commutative but only affine in its first argument. 

\section{Entropy of a composition}

In this section we analyze the behaviour of the entropy of a quantum operation under composition. The bounds~(\ref{boundclass}) on the increase of entropy of probability vectors under discrete dynamics allowed S\l omczy\'nski to prove the subadditivity relation~\cite{Sl02,ZKSS03}
\begin{equation}
 \s H_{\r I}(T_1) 
 \le \s H_{\r I}(T_2 T_1) 
 \le \s H_{\r I}(T_1) + \s H_{\r I}(T_2)  
\label{addclass1}
\end{equation}
provided both stochastic matrices $T_1$ and $T_2$ have the same invariant state, $P^{\r I}_1 = P^{\r I}_2$. Restricting our attention to the case of bistochastic matrices for which $P^{\r I}_1 = P^{\r I}_2 = P_* = \{\frac{1}{N}, \ldots, \frac{1}{N}\}$ we use~(\ref{shanT}) instead of~(\ref{shanTI}) and drop the label ``I'' in the subadditivity relation to get
\begin{equation}
 \s H(T_1) \le \s H(T_2 T_1) \le \s H(T_1) + \s H(T_2) \ . 
\label{addclass2}
\end{equation}
Considering a product of three bistochastic matrices S\l omczy\'nski proved~\cite{Sl02} a strong subadditivity relation for classical dynamics,
\begin{equation}
 \s H(T_3 T_2 T_1) + \s H(T_2 ) 
 \le \s H(T_3 T_2) + \s H(T_2 T_1) \ . 
\label{addclass3}
\end{equation}

\subsection{Dynamical subadditivity and strong subadditivity for bistochastic maps}

Motivated by the classical results above we formulate and prove their quantum counterparts.

\begin{theorem}[Dynamical subadditivity for bistochastic quantum operations]
\label{thm1} 
Let $\Phi_1$ be a bistochastic quantum operation and $\Phi_2$ a general stochastic quantum map then their entropies satisfy the subadditivity inequality 
\begin{equation}
 \s S(\Phi_2 \circ \Phi_1)  \le \s S(\Phi_1) + \s S(\Phi_2)  \ . 
\label{addquant1}
\end{equation}
If both $\Phi_1$ and $\Phi_2$ are bistochastic then
\begin{equation}
 \max\bigl( \{\s S(\Phi_1), \s S(\Phi_2) \} \bigr) 
 \le \min\bigl( \{\s S(\Phi_1 \circ \Phi_2), 
 \s S(\Phi_2 \circ \Phi_1) \} \bigr) \ . 
\label{addquant3}
\end{equation}
An equivalent statement of~(\ref{addquant3}) is the \co{triangle inequality for composition}
\begin{equation}
 \max\bigl( \{\s S(\sigma_1), \s S(\sigma_2)\} \bigr) 
 \le \min\bigl( \{\s S(\sigma_2 \odot \sigma_1), 
 \s S(\sigma_1 \odot \sigma_2) \} \bigr) 
 \le \s S(\sigma_1) + \s S(\sigma_2)\ ,  
\label{addquant1b}
\end{equation}
where $\sigma_1,\sigma_2 \in \c D^{\r{II}}_{N^2}$ and the set $\c D^{\r{II}}_{N^2}$ has been defined in~(\ref{subbist}).
\end{theorem}

\begin{proof}
$i)$ We first show the upper bound~(\ref{addquant1}).

Let $\Theta$ be a stochastic quantum operation, generally not bistochastic, with Kraus form
\begin{equation} 
 \Theta(\rho) = \sum_{\alpha=1}^K C_\alpha \rho\, C_\alpha^\dagger \ .
\end{equation}
Introduce a map from $\c H_N$ to $\c H_N \otimes \c H_N \otimes \c H_K$ as follows: fix an orthonormal basis $\{|\alpha\>\}$ in $\c H_K$ and let
\begin{equation}
 F\varphi := \sum_{\alpha=1}^K \bigl( C_\alpha\varphi \bigr) \otimes|\alpha\>\ ,
\end{equation}
The adjoint map acts as
\begin{equation}
 F^\dagger \varphi \otimes |\alpha\> = C_\alpha^\dagger \varphi
\end{equation}
and one checks that
\begin{equation}
 F^\dagger F \varphi 
 = F^\dagger \sum_{\beta=1}^K \bigl( C_\beta \varphi \bigr) \otimes |\beta\>
 = \sum_{\beta=1}^K C_\beta^\dagger C_\beta \varphi
 = \varphi \ ,
\end{equation}
since $\Theta$ is trace preserving. 
Therefore $F$ is an isometry. It follows in particular that, for an arbitrary $N$-dimensional matrix $\rho$, $F\,\rho\,F^\dagger$ and $\rho$ have up to multiplicities of zero the same eigenvalues.   

Using $F$ we can express the Lindblad operator $\omega$, see~(\ref{omega}), as
\begin{equation}
 \omega 
 = F\, \rho F^\dagger
 = \sum_{\alpha,\beta=1}^K C_\alpha \rho\, C_\beta^\dagger \otimes 
 |\alpha\> \<\beta| \ .
\label{wromega}
\end{equation}
The operator $\omega$ is a density matrix on the composite system $\c H_N \otimes \c H_K$. If the initial state is maximally mixed, $\rho=\rho_{*}$, taking the partial trace over the first subsystem we obtain the following density matrix of the ancilla $E = \c H_K$, 
\begin{equation}
 \rho
 = \tr_N \omega
 = \hat\sigma(\Theta,\rho_*) 
 = \varsigma
 = \frac{1}{N}\, D_\Theta \ . 
\end{equation}

We perform the construction of above for the composition $\Phi_2 \circ \Phi_1$. We first write both quantum operations $\Phi_1$ and $\Phi_2$ in their Kraus forms
\begin{equation} 
 \Phi_1(\rho) = \sum_{\alpha=1}^M A_\alpha \rho\, A_\alpha^\dagger 
 \qquad\text{and}\qquad 
 \Phi_2(\rho) = \sum_{\alpha=1}^M B_\alpha \rho\, B_\alpha^\dagger
\end{equation}
where $M = N^2$. Putting $\Theta = \Phi_2 \circ \Phi_1$ we consider its Kraus decomposition
\begin{equation}
 C_{\alpha_2\alpha_1} := B_{\alpha_2} A_{\alpha_1}
\end{equation}
implying that $K = N^2 \times N^2 = N^4$. Consider a state $\omega_{21}$ acting on an extended tripartite space $\c H_N \otimes \c H_M \otimes \c H_M = \c H_N \otimes E_2 \otimes E_1$ 
\begin{align}
 \omega_{21} 
 &= \sum_{\alpha_1, \alpha_2, \beta_1, \beta_2 = 1}^N 
 C_{\alpha_2\alpha_1} \rho\, C_{\beta_2\beta_1}^\dagger \otimes
 |\alpha_2 \otimes \alpha_1\> \<\beta_2 \otimes \beta_1| 
\nonumber \\
 &= \sum_{\alpha_1, \alpha_2, \beta_1, \beta_2 = 1}^N 
  B_{\alpha_2} A_{\alpha_1} \rho\,  \bigl( B_{\beta_2} A_{\beta_1} \bigr)^\dagger 
 \otimes |\alpha_2 \otimes \alpha_1\> \<\beta_2 \otimes \beta_1| 
\label{stop}
\end{align}
and let $\rho_{21}$ be the restriction of $\omega_{21}$ to the ancilla $E_2 E_1 = \c H_M \otimes \c H_M$. We also compute the restrictions $\rho_1$ and $\rho_2$ to the first and second ancilla. Assuming that $\rho=\rho_{*}=\frac{1}{N}\idty$ we obtain:
\begin{align}
 &\rho_{12} 
 = \frac{1}{N}\, D_{\Phi_2 \circ \Phi_1} \\
 &\rho_1 
 = \tr_{\c H_N \otimes E_2} \omega_{21} 
 = \tr_{\c H_N} \frac{1}{N} \sum_{\alpha_1,\beta_1 = 1}^N 
   A_{\alpha_1} \idty A_{\beta_1}^\dagger \otimes
  |\alpha_1\> \<\beta_1|
 = \frac{1}{N}\, D_{\Phi_1} \\
 &\rho_2
 = \tr_{\c H_N \otimes E_1} \omega_{21} 
 = \tr_{\c H_N} \frac{1}{N} \sum_{\alpha_2,\beta_2 = 1}^N 
   B_{\alpha_2} \idty B_{\beta_2}^\dagger \otimes
  |\alpha_2\> \<\beta_2|
 = \frac{1}{N}\, D_{\Phi_2}\label{rhodwa} \ .
\end{align}
The last inequality holds because $\Phi_1$ is bistochastic. The upper bound of~(\ref{addquant1}) now follows by applying subadditivity of the entropy to the state $\rho_{12}$, see~\cite{ohypet,BZ06}.

\medskip
$ii)$ The lower bound is a special case of the bound in Theorem~\ref{thm2}. As $\Phi_1$ and $\Phi_2$ are bistochastic
\begin{equation}
 \Phi_1(\rho_*) = \Phi_2(\rho_*) = \rho_*
\end{equation}
and the contribution of the terms within the square brackets vanishes.
\end{proof}

\begin{corollary}
 Let $\Phi$ be a bistochastic map, then
\begin{equation}
 \s S\bigl( \Phi^{\circ n} \bigr) \le n\, \s S(\Phi)\ .
\label{cor}
\end{equation}
\end{corollary}

It should be remarked that~(\ref{cor}) is only meaningful when $\s S(\Phi) \ll \ln N$ and $n$ is not too large as the entropy of a map is anyway bounded by $2 \ln N$. 

In a similar way one may analyze properties of a concatenation of three consecutive operations. Motivated by the inequality~(\ref{addclass3}) for the entropy of the products of three bistochastic matrices, we formulate its quantum mechanical counterpart.

\begin{theorem}[Strong dynamical subadditivity for bistochastic 
quantum operations] 
Let $\Phi_1$, $\Phi_2$ and $\Phi_3$ be bistochastic quantum operations then their entropies satisfy the inequality 
\begin{equation}
 \s S(\Phi_3 \circ \Phi_2 \circ \Phi_1 ) + \s S(\Phi_2) 
 \le \s S(\Phi_3 \circ \Phi_2 ) + \s S(\Phi_2 \circ \Phi_1 ) \ .
\label{strongb}
\end{equation}
This is equivalent to 
\begin{equation}
 \s S(\sigma_3 \odot \sigma_2 \odot \sigma_1) 
 + \s S(\sigma_2)
 \le \s S(\sigma_3 \odot \sigma_2) 
 + \s S(\sigma_2 \odot \sigma_1)\ ,  
\end{equation}
where $\sigma_1, \sigma_2, \sigma_3 \in \c D^{\r{II}}_{N^2}$.
\end{theorem}

\begin{proof}
Consider three bistochastic quantum operations $\Phi_1$, $\Phi_2$ and $\Phi_3$ acting in sequence on the maximally mixed state $\rho_*\in \c M_N$. We repeat the construction of the Linblad operator as in~(\ref{wromega}) but now for three ancillas $E_1$, $E_2$ and $E_3$ and obtain a density matrix $\omega_{321}$ acting on $\c H_N \otimes E_3 \otimes E_2 \otimes E_1$. The restrictions of $\omega_{321}$ to some of the ancillas will be denoted as before by $\rho$'s. We now write the strong subadditivity of quantum entropy for the system $E_3E_2E_1$~\cite{lierus}
\begin{equation}
 \s S(\rho_{321}) + \s S(\rho_2) \le \s S(\rho_{21}) + \s S(\rho_{32}) \ .
\label{wrssat}
\end{equation}
This yields, using the bistochasticity of $\Phi_3$ and $\Phi_2$
\begin{align}
 &\s S\bigl( \rho_2 \bigr) = \s S\bigl( \Phi_2 \bigr) \\
 &\s S\bigl( \rho_{21} \bigr) = \s S\bigl( \Phi_2 \circ \Phi_1 \bigr) \\
 &\s S\bigl( \rho_{32} \bigr) = \s S\bigl( \Phi_3 \circ \Phi_2 \bigr) \\
 &\s S\bigl( \rho_{321} \bigr) 
 = \s S\bigl( \Phi_3 \circ \Phi_2 \circ \Phi_1 \bigr) \ ,
\label{wrinequ3}
\end{align}
which ends the proof.
\end{proof}
 
\subsection{Generalization of dynamical subadditivity for stochastic maps}

Results obtained in previous section for bistochastic maps can be generalize for the case of arbitrary stochastic maps.
 
\begin{theorem}[Dynamical subadditivity for quantum operations]
\label{thm2} 
Let $\Phi_1$ and $\Phi_2$ be quantum operations then their 
entropies satisfy the inequalities 
\begin{align}
 \s S(\Phi_1) + \Delta_{1}
 \le \s S(\Phi_2 \circ \Phi_1) \le \s S(\Phi_{1}) + \s S(\Phi_{2}) + \Delta_{2}\label{addquant2}
\end{align}
where
\begin{align}
&\Delta_{1} = \s S\bigl(\Phi_2 \circ \Phi_1(\rho_*) \bigr) 
 - \s S\bigl(\Phi_1(\rho_*) \bigr)\ ,\\
&\Delta_{2} = \s S\Bigl(\hat\sigma\bigl(\Phi_{2}, \Phi_{1}(\rho_{*})\bigr)\Bigr)-\s S(\Phi_{2}) .
\end{align}
\end{theorem}

\begin{proof}
We repeat the proof of Theorem~\ref{thm1} for the composition  $\Phi_2 \circ \Phi_1$ up to the construction of the state $\omega$ on $\c H_N \otimes E_2 \otimes E_1$, see~(\ref{wromega}). If the first operation is not bistochastic then the restriction (\ref{rhodwa}) has the form
\begin{equation}
\rho_2
 = \tr_{\c H_N \otimes E_1} \omega_{21} 
 = \tr_{\c H_N} \sum_{\alpha_2,\beta_2 = 1}^N 
   B_{\alpha_2} \Phi_{1}(\rho_{*})\, B_{\beta_2}^\dagger \otimes
  |\alpha_2\> \<\beta_2|
 = \hat\sigma\bigl(\Phi_{2}, \Phi_{1}(\rho_{*})\bigr)\ .
\end{equation}
The restrictions $\rho_{12}$ and $\rho_{1}$ are the same as before. The entropy of $\hat\sigma$ forms the first term of $\Delta_{2}$ so subadditivity   of the entropy of $\rho_{12}$ implies the upper bound of (\ref{addquant2}). 

To prove the lower bound consider (\ref{stop}). As an initial state let us take $\rho_{*}=\frac{1}{N}\idty$ and purify it by adding an additional party $R$ of dimension $N$. This yields a pure state on $R \otimes \c H_N \otimes E_2 \otimes E_1$ defined by the normalized vector
\begin{equation}
 \frac{1}{N} \sum_{i=1}^N \sum_{\alpha_1,\alpha_2 = 1}^{N^2}
 |i\> \otimes \bigl( B_{\alpha_2}\, A_{\alpha_1} |i\> \bigr)
 \otimes |\alpha_2\> \otimes |\alpha_1\> \ .
\end{equation}
We now use the strong subadditivity of the entropy denoting by $\rho_R$ the restriction of this pure state to the party $R$ with similar notations for restrictions to other parties
\begin{equation}
 \s S(\rho_{RE_2E_1}) + \s S(\rho_{E_1}) 
 \le \s S(\rho_{RE_1}) + \s S(\rho_{E_2E_1})\ .
\label{wrssa}
\end{equation}
Computing all the terms appearing in~(\ref{wrssa}),
\begin{align}
 &\s S\bigl( \rho_{E_1} \bigr) = \s S\bigl( \Phi_1 \bigr) \\
 &\s S\bigl( \rho_{E_2E_1} \bigr) = \s S\bigl( \Phi_2 \circ \Phi_1 \bigr) \\
 &\s S\bigl( \rho_{RE_1} \bigr) = \s S\bigl( \Phi_1(\rho_*) \bigr) \\
 &\s S\bigl( \rho_{RE_2E_1} \bigr) = \s S\bigl( \Phi_2 \circ \Phi_1 (\rho_*) \bigr)
\end{align}
we arrive at the lower bound (\ref{addquant2}).
\end{proof}

The above inequalities formulated in the language of quantum maps obviously hold for the composition of states which belong to the set~(\ref{substoch}) and its subset~(\ref{subbist}).

The above results may be linked to properties of $coherent\ information$ defined by Schumacher and Nielsen \cite{SN96} as a function of exchange entropy $\s S\bigl(\hat\sigma(\Phi, \rho)\bigr)$:
\begin{align}
I(\Phi, \rho)=\s S\bigl(\Phi(\rho)\bigr)-\s S\bigl(\hat\sigma(\Phi, \rho)\bigr).
\end{align}
Denoting by $I_{1}$ the coherent information for the first operation, $I_{1}=\s S\bigl(\Phi_{1}(\rho)\bigr)-\s S\bigl(\hat\sigma(\Phi_{1}, \rho)\bigr)$, and by $I_{21}$ the analogous quantity for the concatenation $I_{21}=\s S\Bigl(\Phi_{2}\bigl(\Phi_{1}(\rho)\bigr)\Bigr)-\s S\bigl(\hat\sigma(\Phi_{2}\circ\Phi_{1}, \rho)\bigr)$, these authors proved \cite{SN96} the $quantum\ data\ processing\ inequality$:
\begin{align}
I_{21}\le I_{1}\ .
\end{align} 
This result implies the lower bound of (\ref{addquant2}).

\subsection{Entropy of a product of stochastic matrices}

Restricting our attention to the case of diagonal dynamical matrices we may analyze classical analogues of the above results. In this case an application of a quantum map corresponds to action of a stochastic matrix $T$ on a classical probability vector $P$.

Theorem \ref{thm2} implies that in the classical case for an arbitrary stochastic matrices $T_{1}$ and $T_{2}$ the following bounds hold,
\begin{equation}
 \s H(T_1) + \Bigl[\s H\bigl( T_2 T_1(P_*) \bigr) 
 - \s H\bigl(T_1(P_*) \bigr) \Bigr]
 \le \s H(T_2 T_1) \le  \s H(T_{1}) + \s H_{T_{1}P_{*}}(T_{2}) + \s H(T_{1}P_{*})\ ,   
\label{addquant2b}
\end{equation}
where $P_{*}=\frac{1}{N}(1,...,1)$.

However it is possible to obtain a stronger upper bound by using properties of strong subadditivity of entropy. Inequality (\ref{wrssat}) is equivalent to an inequality for the exchange entropy,
\begin{equation}
\s S\bigl(\hat\sigma(\Phi_{3}\circ\Phi_{2}\circ\Phi_{1}, \rho_{*})\bigr)+\s S\Bigl(\hat\sigma\bigl(\Phi_{2}, \Phi_{1}(\rho_{*})\bigr)\Bigr) \le \s S\bigl(\hat\sigma(\Phi_{2}\circ\Phi_{1}, \rho_{*})\bigr)+\s S\Bigl(\hat\sigma\bigl(\Phi_{3}\circ\Phi_{2}, \Phi_{1}(\rho_{*})\bigr)\Bigr)\ .
\label{ssaee}
\end{equation}
By restriction to the classical case we obtain an inequality for three stochastic matrices $T_{1}=T_{1}(\Phi_{1})$, $T_{a}=T_{a}(\Phi_{2})$ and $T_{2}=T_{2}(\Phi_{3})$ (note the notation chosen, which is convenient to state the final result). This inequality is similar to the strong subadditivity, but different weights for entropies are used. By substitution $T_{a}=\idty$ we arive with a result analogous to (\ref{addquant2}).

\begin{theorem} 
Let $T_{1}$ and $T_{2}$ be arbitrary stochastic matrices. Then the entropy of their product is bounded by
\begin{align}
\s H(T_{1})+\delta_{1} \le \s H(T_{2}T_{1}) \le \s H(T_{2})+\s H(T_{1})+\delta_{2}
\end{align}
where
\begin{align}
& \delta_{1}=\s H(T_{2}T_{1}(P_{*}))-\s H(T_{1}(P_{*})),\\
& \delta_{2}=\s H_{T_{1}P_{*}}(T_{2})-\s H(T_{2})\ .
\end{align} 
\end{theorem}

This classical version of (\ref{addquant2}) valid for arbitrary stochastic matrices $T_{1}$ and $T_{2}$ can be considered as a direct generalization of the result of S{\l}omczy{\'n}ski~\cite{Sl02} obtained for bistochastic matrices.
If both matrices $T_1$ and $T_2$ are bistochastic a complementary lower bound for the entropy of their product holds $\s H(T_2) \le \s H(T_2 T_1)$, so in this case~(\ref{addclass2}) can be rewritten in the stronger symmetric form~\cite{ZKSS03} 
\begin{equation}
 \max\bigl( \{\s H(T_1), \s H(T_2) \} \bigr) 
 \le \min\bigl( \{\s H(T_1 T_2), \s H(T_2 T_1) \} \bigr) \ . 
\label{addclass4}
\end{equation}

\subsection{Dynamical subadditivity of quasi-free bistochastic quantum operations}

The general results of Theorems~\ref{thm1} and~\ref{thm2} can of course be applied to the case of quasi-free completely positive maps. It is however also possible to derive them directly for such special maps bypassing e.g.\ the strong subadditivity that was used in the general proof. By way of illustration we prove Theorem~\ref{thm1} for this class of maps. 

Let $R_1$ and $R_2$ be $N$-dimensional matrices such that $\norm{R_1} \le 1$ and $\norm{R_2} \le 1$. Remark first that 
\begin{equation}
 \Phi_{R_2} \circ \Phi_{R_1} = \Phi_{R_2 R_1} \ . 
\end{equation}
As 
\begin{equation}
 \abs{R_2R_1}^2 = R_2^*R_1^*R_1R_2 \le R_2^*R_2 = \abs{R_2}^2 
\end{equation} 
and as 
\begin{equation}
 x\in[0,1] 
 \mapsto \eta\bigl( {\textstyle\frac{1}{2}}\, (1+x) \bigr) 
 + \eta\bigl({\textstyle\frac{1}{2}}\, (1-x) \bigr)
\end{equation}
is monotonically decreasing, we obtain
\begin{equation}
 \s S\bigl( \Phi_{R_2} \bigr) 
 \le \s S\bigl( \Phi_{R_2 R_1} \bigr) 
 = \s S\bigl( \Phi_{R_2} \circ \Phi_{R_1} \bigr) \ .
\end{equation}
Using $\s S(\Phi_R) = \s S(\Phi_{R^*})$ allows to exchange the roles of $R_1$ and $R_2$ and so
\begin{equation}
 \max\Bigl( \bigl\{ \s S\bigl( \Phi_{R_1} \bigr), \s S\bigl( \Phi_{R_2} \bigr)\bigr\}\Bigr)
 \le \s S\bigl( \Phi_{R_2} \circ \Phi_{R_1} \bigr) \ .
\end{equation}

If we denote by $\{\lambda_1(R), \lambda_2(R), \ldots, \lambda_N(R)\}$ the singular values of the $N$-dimensional matrix $R$ arranged in decreasing order, then
\begin{align}
 &\lambda_1(R_2R_1) 
\nonumber \\
 &\ge \max\Bigl( \bigl\{ \lambda_j(R_2)\, \lambda_{N-j+1}(R_1) \,:\, 1 \le j \le N\bigr\}\Bigr)
\nonumber\\ 
 &\lambda_1(R_2R_1) + \lambda_2(R_2R_1)  
\nonumber \\
 &\ge \max\Bigl( \bigl\{ \lambda_{j_1}(R_2)\, \lambda_{N-j_1+1}(R_1) +
 \lambda_{j_2}(R_2)\, \lambda_{N-j_2+1}(R_1) \,:\, 1 \le j_1 < j_2 \le N\bigr\}\Bigr)
\nonumber\\ 
 &\vdots
\end{align}
It follows that 
\begin{align}
 &\bigl\{ {\textstyle\frac{1}{2}}\, (1 + \lambda_j(R_2R_1)) \,:\, 1 \le j \le N \bigr\}
 \bigcup \bigl\{ {\textstyle\frac{1}{2}}\, (1 - \lambda_j(R_2R_1)) \,:\, 1 \le j \le N
 \bigr\} 
\nonumber\\
\intertext{is less mixed than}
 &\bigl\{ {\textstyle\frac{1}{2}}\, (1 + \lambda_j(R_2)\, \lambda_{N-j+1}(R_1)) \,:\, 1 \le j \le N \bigr\}
 \bigcup \bigl\{ {\textstyle\frac{1}{2}}\, (1 - \lambda_j(R_2)\, \lambda_{N-j+1}(R_1)) \,:\, 1 \le j \le N
 \bigr\} \ .
\end{align}
This implies by concavity of $\eta$ that
\begin{equation}
 \s S(\Phi_{R_2R_1}) 
 \le 2 \sum_{j=1}^N \Bigl( \eta\bigl( {\textstyle\frac{1}{2}}\, (1 +
 \lambda_j(R_2)\, \lambda_{N-j+1}(R_1)) \bigr) 
 + \eta\Bigl( {\textstyle\frac{1}{2}}\, (1 - \lambda_j(R_2)\,
 \lambda_{N-j+1}(R_1)) \Bigr)\Bigr) \ .
\end{equation}
Finally observing that for $0 \le a,b \le 1$
\begin{equation}
 \Bigl\{ {\textstyle\frac{1}{2}}\, (1 + ab), {\textstyle\frac{1}{2}}\,
 (1 - ab), 0, 0 \Bigr\}
\end{equation}
is less mixed than
\begin{equation}
 \Bigl\{ {\textstyle\frac{1}{2}}\, (1 + a), {\textstyle\frac{1}{2}}\,
 (1 - a) \Bigr\} 
 \times \Bigl\{ {\textstyle\frac{1}{2}}\, (1 + b), {\textstyle\frac{1}{2}}\,
 (1 - b) \Bigr\} 
\end{equation}
we obtain
\begin{equation}
 \s S(\Phi_{R_2} \circ \Phi_{R_1}) 
 \le \s S(\Phi_{R_1}) + \s S(\Phi_{R_2}) \ .
\end{equation}

\section{Concluding remarks}

In this paper we introduced the composition $\odot$ between states on a bipartite system and derived some basic properties. This composition reflects the concatenation of quantum operations under the \jam\ isomorphism. Next, we introduced a simple notion of entropy of a quantum map in order to quantify its randomizing properties. We proved the property of subadditivity for a composition of arbitrary bistochastic maps and found its generalization for the case of stochastic quantum maps. The connection between maps and states allows to formulate these properties purely in terms of states and the $\odot$ composition. A restriction to the classical setting leads to a generalization of recently obtained bounds on the entropy of a product of two bistochastic matrices~\cite{Sl02} to the case of a product of arbitrary stochastic matrices.

Recently~\cite{wolcir}, the concatenation of quantum maps has been investigated from the point of view of divisibility. The authors consider the determinant of a super-operator instead of the entropy and show that it is contractive with respect to composition.

A more detailed understanding of the randomizing properties of a quantum map should be provided by constructing a Markov like process. A possible track is to generate a state on a spin half-chain in the spirit of finitely correlated states~\cite{fannacwer} and to consider the associated dynamical entropy. This is a subject of future research.  

\medskip
\noindent
\textbf{Acknowledgements:}
We enjoyed fruitful discussions with I.~Bengtsson, V.~Cappellini and W.~S{\l}omczy{\'n}ski. We acknowledge financial support by the bilateral project BIL05/11 between Flanders and Poland, by the Polish Ministry of Science and Information Technology under the grant 1\, P03B\, 042\, 26
and by the European Research Project SCALA.

\appendix
\section{Appendix}

The aim of this appendix is to make for low dimensional one-particle spaces the connection between quasi-free states the usual density matrix description of a state more explicit. The CAR algebra $\c A_N := \c A(\c H_N)$ is isomorphic to the algebra of matrices of dimension $2^N$. Because of the finite dimensionality all representations are equivalent. A particularly useful representation is on Fermionic Fock space
\begin{equation}
 \Gamma(\c H_N) := \Cx \oplus \c H_N \oplus \c H_N^{\,2,\, \r{as}} \oplus \cdots \oplus \Cx \ .
\end{equation}
Here $\c H_N^{\,k,\, \r{as}}$ denotes the subspace of antisymmetric vectors under the action of the permutation group of $\c H_N^{\otimes k}$. A quasi-free state $\<\ \>_Q$ with $0 \le Q < \idty$ corresponds to a density matrix of the form
\begin{equation} 
 \det(\idty -Q)\ \Bigl\{ 1 \oplus \frac{Q}{\idty - Q} 
 \oplus \Bigl( \frac{Q}{\idty - Q} \otimes \frac{Q}{\idty - Q}\Bigr)\Bigr|_{\c H_N^{\, 2,\, \r{as}}} 
 \oplus \cdots \Bigr\} \ .
\end{equation}
The pure quasi-free states are limiting cases of this formula, their density matrices are one-dimensional projectors of the form
\begin{equation}
 0 \oplus \cdots \oplus \Bigl( \underbrace{P \otimes \cdots \otimes P}_{k\text{-times}} \Bigr)\Bigr|_{\c H_N^{\, k,\, \r{as}}} \oplus \cdots
\end{equation}
Here $P$ is a $k$-dimensional projector acting on $\c H_N$, actually the symbol of the corresponding state.

For $N = 1$, the algebra $\c A_1$ is isomorphic to $\c M_2$. Any state $\rho$ on $\Cx^2$ can be mapped on a quasi-free state by choosing a suitable identification of $\c M_2$ with $\c A_1$. Once this identification is fixed, the convex set of quasi-free states is equal to the set of diagonal density matrices (in the basis diagonalizing $\rho$).

$\c A_2$ is isomorphic to $\c M_4$. One can show that any density matrix on $\Cx^4$ which has the property that the product of its largest and smallest eigenvalues is equal to the product of the two others can be mapped onto a quasi-free state. Once this identification fixed, the convex hull of the set of quasi-free states equals 
\begin{equation}
 \{ \lambda_1 \oplus \lambda_2\, \rho \oplus \lambda_3 \,:\, (\lambda_1, \lambda_2, \lambda_3) \text{ a probability vector and $\rho$ a 2D density matrix } \}\ .
\end{equation}    
This is exactly the set of block diagonal density matrices in a decomposition $\Cx^4 = \Cx \oplus \Cx^2 \oplus \Cx$. 

For $N = 3$, the convex hull of quasi-free states is also seen to be the block diagonal density matrices in a decomposition $\Cx^8 = \Cx \oplus \Cx^3 \oplus \Cx^3 \oplus \Cx$. 

For $N = 4$, and higher dimensions, things change. It turns out that the convex hull of quasi-free states is now a strict subset of  the block diagonal density matrices in a decomposition $\Cx^{16} = \Cx \oplus \Cx^4 \oplus \Cx^6 \oplus \Cx^4 \oplus \Cx$. The pure quasi-free states with support in the $\Cx^6$ term correspond to particular 1D projectors, namely restrictions of projectors the form $P \otimes P$ with $P$ a 2D projector on $\Cx^4$ to the antisymmetric subspace of $\Cx^4 \otimes \Cx^4 \approx \Cx^6$. Such $P$ define a 8D real manifold, while the 1D projectors on $\Cx^6$ form a 10D real manifold. Nevertheless the convex hull of the $P \otimes P$ has a non-zero volume. In fact any density matrix on $\Cx^6$ may be obtained by a linear, generally not convex, combination of pure quasi-free states supported in $\Cx^6$.   


\end{document}